\documentclass[11pt,a4paper]{article}
\usepackage[utf8]{inputenc}
\usepackage[T1]{fontenc}
\usepackage{amsmath}
\usepackage{amsfonts}
\usepackage{amssymb}
\usepackage{amsthm}
\usepackage[
left=3.4cm, right=3.4cm, top=3.1cm, bottom=3.1cm
]{geometry}
\usepackage{authblk}
\usepackage[inline]{enumitem}
\usepackage{multicol}
\usepackage{stmaryrd}
\usepackage{bm}
\usepackage{mathrsfs}
\usepackage{textcomp} 
\usepackage{framed} 
\usepackage{hyperref}

\usepackage{tikz}
		\usetikzlibrary{positioning,fit,shapes}

\newcommand{\leftb}{\scalebox{0.9}{\normalfont\texttt{[}}}
\newcommand{\rightb}{\scalebox{0.9}{\normalfont\texttt{]}}}

\newcommand{\B}[1]{\leftb #1 \rightb}

\newcommand{\To}{\Rightarrow}

\newcommand{\lto}{\hookrightarrow}

\newcommand{\cod}{\mathsf{c}}
\newcommand{\ant}{\mathsf{a}}

\newcommand{\e}{\mathsf{e}} 
\newcommand{\tra}{Tr}
\newcommand{\val}[1]{\B{#1}}
\newcommand{\vval}[1]{\leftb\!\!\leftb #1 \rightb\!\!\rightb}

\newtheorem{theorem}{Theorem}
\newtheorem{lemma}{Lemma}
\newtheorem{proposition}{Proposition}

\newtheorem{definition}{Definition}
	\theoremstyle{definition}

	\theoremstyle{remark}

\title{Embedding Kozen-Tiuryn Logic into Residuated One-Sorted Kleene Algebra with Tests\thanks{
Preprint of: I. Sedlár and J.J.~Wannenburg:  Embedding Kozen-Tiuryn Logic into Residuated One-Sorted Kleene Algebra with Tests. In: A. Ciabattoni, E. Pimentel, R. J. G. B. de Queiroz (Eds.): \emph{Proc. 28th Int. Conference of Logic, Language, Information, and Computation (WoLLIC 2022)}, pp. 221-236. LNCS 13468. Springer, 2022. (Available online at \url{https://link.springer.com/chapter/10.1007/978-3-031-15298-6_14}.)
}}
\author{Igor Sedl\'ar and Johann  J.~Wannenburg}
\affil{The Czech Academy of Sciences, Institute of Computer Science 
	\authorcr Pod Vod\'arenskou v\v{e}\v{z}\'i 271/2, Prague, The Czech Republic}
	
	\date{ }
	
\begin{document}
\maketitle

\vspace*{-7mm}

\begin{abstract}
Kozen and Tiuryn have introduced the substructural logic $\mathsf{S}$ for reasoning about correctness of while programs (ACM TOCL, 2003). The logic $\mathsf{S}$ distinguishes between tests and partial correctness assertions, representing the latter by special implicational formulas. Kozen and Tiuryn's logic extends Kleene altebra with tests, where partial correctness assertions are represented by equations, not terms. Kleene algebra with codomain, $\mathsf{KAC}$, is a one-sorted alternative to Kleene algebra with tests that expands Kleene algebra with an operator that allows to construct a Boolean subalgebra of tests. In this paper we show that Kozen and Tiuryn's logic embeds into the equational theory of the expansion of $\mathsf{KAC}$ with residuals of Kleene algebra multiplication and the upper adjoint of the codomain operator. 

\noindent \paragraph{Keywords:} Hoare logic, 
Kleene algebra with codomain, 
Kleene algebra with tests, 
Partial correctness, 
Substructural logic
\end{abstract}

\section{Introduction}

Kleene algebra with tests \cite{Kozen1997}, \textsf{KAT}, is a simple algebraic framework for verifying properties of propositional while programs. \textsf{KAT} is two-sorted, featuring a Boolean algebra of tests embedded into a Kleene algebra of programs. The equational theory of \textsf{KAT} is PSPACE-complete \cite{CohenEtAl1996} and \textsf{KAT} subsumes Propositional Hoare logic, \textsf{PHL} \cite{Kozen2000}. A \textsf{PHL} partial correctness assertion  $\{ b \} \,p\, \{ c \}$, meaning that $c$ holds after each terminating execution of program $p$ in a state satisfying $b$, is represented in \textsf{KAT} by the equation $bp\bar{c} = 0$ or, equivalently, $bp = bpc$. The universal Horn theory (or the quasi-equational theory) of \textsf{KAT} consisting of quasi-equations of the form $E \longrightarrow p = q$ where $E$ is a finite set of equations of the form $r = 0$ embeds into the equational theory of \textsf{KAT} \cite{KozenSmith1997} and is therefore PSPACE-complete as well. 

Kozen and Tiuryn \cite{KozenTiuryn2003} introduce a substructural logic \textsf{S} that extends \textsf{KAT} and represents partial correctness assertions as implicational formulas $bp \To c$. They argue that the implicational rendering of partial correctness assertions has certain advantages over the equational one, e.g.~it facilitates a better distinction between local and global properties. \textsf{PHL} embeds into \textsf{S} and \textsf{S} is PSPACE-complete \cite{Kozen2001a}. Kozen and Tiuryn prove that \textsf{S} induces a left-handed \textsf{KAT} structure on programs. However, a more thorough discussion of the relation of \textsf{S} to residuated extensions of Kleene algebra, such as Pratt's action logic \cite{Pratt1991}, or to substructural logics in general is not provided. We believe that a deeper investigation of these connections would shed more light on the landscape of program logics and algebras.

In this paper we prove an embedding result that contributes to this. We note that a simple extension of \textsf{KAT} with a residual $\to$ of the Kleene algebra multiplication does not enable an obvious embedding of \textsf{S}: implication formulas of \textsf{S} are test-like in the sense that they entail the multiplicative unit $1$, but terms $p \to b$ of residuated \textsf{KAT} are not test-like, even if $b$ is. To avoid this problem, we work with Kleene algebra with codomain, \textsf{KAC} \cite{DesharnaisEtAl2006,DesharnaisStruth2011}, a one-sorted alternative to \textsf{KAT} featuring a test-forming codomain operator $\cod$. In particular, we introduce an extension of \textsf{KAC} with both residuals of Kleene algebra multiplication and the upper adjoint $\e$ of the codomain operator, which we call \textsf{SKAT}. We show that $\cod (p \to \e(b))$ behaves like $p \To b$ of \textsf{S}. Our main technical result is that \textsf{S} embeds into the equational theory of $\mathsf{SKAT}^{*}$, the class of algebras in \textsf{SKAT} that are based on $^{*}$-continuous Kleene algebras.

The paper is structured as follows. Section \ref{sec:KATS} recalls \textsf{KAT} and \textsf{S}. Section \ref{sec:RKAT} introduces a residuated version of \textsf{KAT}. We show that a naive translation of the implication connective of \textsf{S} into the residual operation does not constitute an embedding. In Section \ref{sec:KAC} we recall \textsf{KAC} and in Section \ref{sec:SKAT} we introduce \textsf{SKAT} and \textsf{SKAT}$^{*}$. Section \ref{sec:embedding} establishes the embedding result. Section \ref{sec:conclusion} concludes the paper and lists some interesting open problems.

\section{\textsf{KAT} and \textsf{S}}\label{sec:KATS}

In this section we recall \textsf{KAT} (Sect.~\ref{sec:KATS-1}) and we outline \textsf{S} (Sect.~\ref{sec:KATS-2}).

\subsection{Kleene algebra with tests}\label{sec:KATS-1}

We assume that the reader is familiar with the notion of an \emph{idempotent semiring}.

\begin{definition}
A \emph{Kleene algebra} \cite{Kozen1994} is an idempotent semiring $(K, \cdot, +, 1, 0)$ expanded with an operation $^{*} : K \to K$ such that
\begin{gather}
1 + x x^{*} \leq x^{*}\label{*1}\\
1 + x^{*}x \leq x^{*}\label{*2}\\
y + xz \leq z \implies x^{*}y \leq z\label{*3}\\
y + zx \leq z \implies y x^{*} \leq z\label{*4}
\end{gather}
A Kleene algebra is \emph{$^{*}$-continuous} iff it satisfies the condition
\begin{equation}\label{*cont}
x y^{*}z = \sum_{n \geq 0} xy^{n} z
\end{equation}
where $y^{0} = 1$ and $y^{n+1} = y^{n}y$. $\sum X$ is the supremum of the set $X \subseteq K$ with respect to the partial order induced by Kleene algebra addition $+$.
\end{definition}

It follows from the definition that $x^{*}$ is the least element $z$ such that $1 \leq z$, $xz \leq z $ and $zx \leq z$. A standard example of a Kleene algebra is a relational Kleene algebra where $K$ is a set of binary relations over some set $S$, $\cdot$ is relational composition, $+$ is set union, $^{*}$ is reflexive transitive closure, $1$ is identity on $S$ and $0$ is the empty set. A relational Kleene algebra where $K$ is the power set of some $S \times S$ will be called \emph{fully relational}. Another standard example is the Kleene algebra of regular languages over a finite alphabet, where $^{*}$ corresponds to finite iteration (Kleene star). Relational Kleene algebras and Kleene algebras of regular languages are both $^{*}$-continuous. Part of the definition of $^{*}$-continuous Kleene algebras is the requirement that suprema of sets of elements of the form $x y^{n} z$ exist. Not each Kleene algebra is $^{*}$-continuous \cite{Kozen1990}, but the equational theories of Kleene algebras and $^{*}$-continuous Kleene algebras coincide \cite{Kozen1994}.

\begin{definition}
A \emph{Kleene algebra with tests} \cite{Kozen1997} is a structure $(K, B, \cdot, +, ^{*}, 1, 0, \, \bar{ } \, )$ such that
\begin{itemize}
\item $(K, \cdot, +, ^{*}, 1, 0)$ is a Kleene algebra,
\item $B \subseteq K$, and
\item $(B, \cdot, +, \, \bar{ }\, , 1, 0) $ is a Boolean algebra.
\end{itemize}
A Kleene algebra with tests is $^{*}$-continuous if its underlying Kleene algebra is.
\end{definition}

Every Kleene algebra is a Kleene algebra with tests; take $B = \{ 0, 1 \}$ and define $\bar{0} = 1$, $\bar{1} = 0$ and $\bar{x} = x$ for $x \notin \{ 1, 0 \}$.\footnote{In contrast to the standard formulation of Kleene algebra with tests \cite{Kozen1997}, we consider $\,\bar{\,}\,$ to be defined on all elements of $K$.} A standard example of a ($^{*}$-continuous) Kleene algebra with tests is a relational Kleene algebra expanded with a Boolean subalgebra of the \emph{negative cone}, i.e.~the elements $x \leq 1$, which in the relational case are subsets of the identity relation. The class of Kleene algebras with tests is denoted as $\mathsf{KAT}$ and the class of $^{*}$-continuous Kleene algebras with tests is denoted as $\mathsf{KAT^{*}}$. As in the test-free case, the equational theory of $\mathsf{KAT}$ is identical to the equational theory of $\mathsf{KAT}^{*}$ \cite{KozenSmith1997}.

Kleene algebras with tests are able to represent while programs and facilitate equational reasoning about their partial correctness:
\begin{itemize}
\item $\mathbf{skip} = 1$ 
\item $p ; q = pq$
\item $\textbf{if}\: b \: \textbf{then} \: p \: \textbf{else}\: q$ $= (bp) +(\bar{b}q)$
\item $\textbf{while}\: b\: \textbf{do}\: p$ $= (bp)^{*}\bar{b}$
\item $\{ b \} \,p\, \{ c \}$ corresponds to $bp\bar{c} = 0$
\end{itemize}

\subsection{Substructural logic of partial correctness}\label{sec:KATS-2}

This section outlines the substructural logic \textsf{S} \cite{KozenTiuryn2003} that allows to represent a partial correctness assertion by a formula (term) instead of an equation as in \textsf{KAT}. As with \textsf{KAT}, the logic \textsf{S} is many-sorted. Let $\mathsf{B} = \{ \mathsf{b}_i \mid i \in \omega \}$ be the set of test variables and let $\mathsf{P} = \{ \mathsf{p}_i \mid i \in \omega \}$ be the set of program variables. We define the following sorts of syntactic objects:
\begin{center}
\begin{tabular}{ll}
tests & $b, c := \mathsf{b}_i \mid 0 \mid b \To c$\\[1mm]
programs & $p,q := \mathsf{p}_i \mid b \mid p \oplus q \mid p \otimes q \mid p^{+}$\\[1mm]
formulas & $f,g := b \mid p \To f$\\[1mm]
environments\hspace*{2mm} & $\Gamma, \Delta := \epsilon \mid \Gamma, p \mid \Gamma, f$\\[1mm]
sequents & $\Gamma \vdash f$
\end{tabular}
\end{center}
We define $1 := 0 \To 0$, $\neg b := b \To 0$ and $p^{*} := 1 \oplus p^{+}$. We will sometimes write $pq$ instead of $p \otimes q$ and $\bar{b}$ instead of $\neg b$. Let $\mathsf{E} = \mathsf{B} \cup \mathsf{P}$ and let $Ex_{\mathsf{S}}$, the set of $\mathsf{S}$-expressions, be the union of the sets of formulas, programs and environments.

Kozen and Tiuryn introduce three kinds of semantics for their language: semantics based on guarded strings, traces and binary relations, respectively. We will work only with binary relational semantics.

\begin{definition}
A \emph{Kozen--Tiuryn model} is a pair $M = (W, V)$, where $W$ is a non-empty set and $V: \mathsf{E} \to 2^{W \times W}$ such that $V(\mathsf{b}) \subseteq \mathrm{id}_{W}$ for all $\mathsf{b} \in \mathsf{B}$.

For each Kozen--Tiuryn model $M$, we define the \emph{$M$-interpretation} function $\val{\,}_M : Ex_{\mathsf{S}} \to 2^{W \times W}$ as follows:
\begin{itemize}
\item $\val{\mathsf{b}}_M = V(\mathsf{b})$
\item $\val{\mathsf{p}}_M = V(\mathsf{p})$
\item $\val{0}_M = \emptyset$
\item $\val{b \To c}_M = \{ (s,s) \mid (s,s) \not\in \val{b}_M \text{ or } (s,s) \in \val{c}_M \}$
\item $\val{p \oplus q}_M = \val{p}_M \cup \val{q}_M$
\item $\val{p \otimes q}_M = \val{p}_M \circ \val{q}_M$
\item $\val{p^{+}}_M = \val{p}_M^{+}$
\item $\val{p \To f}_M = \{ (s,s) \mid \forall t. (s,t) \in \val{p}_M \implies (t,t) \in \val{f}_M \}$
\item $\val{\epsilon}_M = \mathrm{id}_{W}$
\item $\val{\Gamma, \Delta}_M = \val{\Gamma}_M \circ \val{\Delta}_M$
\end{itemize}
(Here $^{+}$ denotes transitive closure and $\circ$ denotes relational composition.)

A sequent $\Gamma \vdash f$ is \emph{valid in $M$} iff, for all $s,t \in W$, if $(s,t) \in \val{\Gamma}_M$, then $(t,t) \in \val{f}_M$. 
\end{definition}

Observe that $\val{f}_M \subseteq \mathrm{id}_W$ for all formulas $f$; if $(s,s) \in \val{f}_M$, then we may say that formula $f$ is true in $s$. Note that $\val{bp \To c}_M$ is the set of $(s,s)$ such that, for all $t$, if $(s,s) \in \val{b}_M$ and $(s,t) \in \val{p}_M$, then $(t,t) \in \val{c}_M$. Hence, $bp \To c$ represents a partial correctness assertion: the formula is true in $s$ iff $b$ is true in $s$ and $p$ connects $s$ with a state $t$ only if $c$ is true in $t$. 

Figure \ref{fig:S} shows the sequent proof system for $\mathsf{S}$. A sequent $\Gamma \vdash f$ is \emph{provable in $\mathsf{S}$} iff there is a finite sequence of sequents that ends with $\Gamma \vdash f$ each of which is either of the form (Id) or (I$0$), or is derived from previous sequents using some of the inference rules.

\begin{theorem}[Kozen and Tiuryn]\label{thm:KT-completeness}
$\Gamma \vdash f$ is provable in $\mathsf{S}$ iff $\Gamma \vdash f$ is valid in all models $M$.
\end{theorem}

\begin{figure}
\hspace*{1.5cm}
\begin{minipage}{0.45\linewidth}
\begin{itemize}\itemsep=3mm
\item[(Id)] $b \vdash b$
\item[(TC)] $\dfrac{\Gamma, b, \Delta \vdash f \qquad \Gamma, \bar{b}, \Delta \vdash f}{\Gamma, \Delta \vdash f}$
\item[(R$\To$)] $\dfrac{\Gamma, p \vdash f}{\Gamma \vdash p \To f}$
\item[(I$\otimes$)] $\dfrac{\Gamma, p, q, \Delta \vdash f}{\Gamma, p \otimes q, \Delta \vdash f}$
\item[(I$\oplus$)] $\dfrac{\Gamma, p, \Delta \vdash f \qquad \Gamma, q, \Delta \vdash f}{\Gamma, p \oplus q, \Delta \vdash f}$
\item[(I$^{+}$)] $\dfrac{g, p \vdash f \qquad g, p \vdash g}{g, p^{+} \vdash f}$
\item[(E$^{+}$)] $\dfrac{\Gamma, p^{+}, \Delta \vdash f}{\Gamma, p, \Delta \vdash f}$
\item[(CC$^{+}$)] $\dfrac{\Gamma, p^{+}, \Delta \vdash f}{\Gamma, p^{+}, p^{+}, \Delta \vdash f}$
\end{itemize}
\end{minipage}\qquad
\begin{minipage}{0.45\linewidth}
\begin{itemize}\itemsep=3mm
\item[(I$0$)] $\Gamma, 0, \Delta \vdash f$
\item[(cut)] $\dfrac{\Gamma \vdash g \qquad \Gamma, g, \Delta \vdash f}{\Gamma, \Delta \vdash f}$
\item[(I$\To$)] $\dfrac{\Gamma, p, f, \Delta \vdash g}{\Gamma, p \To f, p, \Delta \vdash g}$
\item[(E$\otimes$)] $\dfrac{\Gamma, p \otimes q, \Delta \vdash f}{\Gamma, p, q, \Delta \vdash f}$ 
\item[(E$\oplus_{1}$)] $\dfrac{\Gamma, p \oplus q, \Delta \vdash f}{\Gamma, p, \Delta \vdash f}$
\item[(E$\oplus_{2}$)] $\dfrac{\Gamma, p \oplus q, \Delta \vdash f}{\Gamma, q, \Delta \vdash f}$
\item[(W$f$)] $\dfrac{\Gamma, \Delta \vdash g}{\Gamma, f, \Delta \vdash g}$
\item[(W$p$)] $\dfrac{\Gamma \vdash f}{p, \Gamma \vdash f}$ 
\end{itemize}
\end{minipage}\caption{The sequent proof system for $\mathsf{S}$.}\label{fig:S}
\end{figure}

\textsf{S} can be seen as a substructural logic \cite{Restall2000}. From that point of view, \textsf{S} is an extension of the Multiplicative-additive Lambek calculus with the transitive closure operator $^{+}$. However, \textsf{S} contains some rules not usual from the substructural logic perspective, namely, the sort-specific weakening rules (W$f$) and (W$p$), and the implication-formula rules (TC) and (I$\To$). It is therefore interesting to inquire how \textsf{S} relates to mainstream substructural logics. The embedding result of Section \ref{sec:embedding} is helpful in this respect. 

\section{Residuated \textsf{KAT}}\label{sec:RKAT}

The crucial difference between \textsf{KAT} and \textsf{S} is that the latter contains an implication operator $\To$ that is used to formalize partial correctness assertions as formulas (terms) of the language. Therefore, one may wonder if \textsf{S} can be thought of as \textsf{KAT} expanded with a residual of Kleene algebra multiplication. In this section we introduce residuated \textsf{KAT}, \textsf{RKAT}, and we show that a naive mapping from \textsf{S} to \textsf{RKAT} is not an embedding. 

We note that residuated Kleene algebras with tests were studied also in \cite{Jipsen2004}. In that framework, however, it is assumed by definition that the negative cone of the underlying Kleene algebra forms a Boolean algebra of tests. Kozen \cite{Kozen1997} argues, however, that this approach is not entirely satisfactory: tests are usually easily decidable (unlike other propositions, such as halting assertions), and not each Kleene algebra has a Boolean negative cone. Our definition of residuated Kleene algebras with tests adds residuals to \textsf{KAT} defined as usual.      

\begin{definition}
A \emph{residuated idempotent semiring} is a structure of the form $(K, \cdot, +, \to, \lto, 1, 0)$ where
\begin{itemize}
\item $(K, \cdot, +, 1, 0)$ is an idempotent semiring, and
\item $\to, \lto$ are binary operations on $K$ satisfying the \emph{residuation laws}
\begin{equation}
xy \leq z \iff x \leq y \to z \iff y \leq x \lto z
\end{equation}
\end{itemize}
\end{definition}

In residuated idempotent semirings, $\to$ is called the \emph{right residual} of $\cdot$ and $\lto$ is called the \emph{left residual}. An idempotent semiring is called \emph{right-residuated} (\emph{left-residuated}) if it is a reduct of a residuated idempotent semiring with only the right (left) residual. Note that residuated idempotent semirings necessarily have a top element: $x0 \leq 0$ implies $x \leq 0 \to 0$ and $0x \leq 0$ implies $x \leq 0 \lto 0$.

The study of residuated semirings has some precedent. Pratt adds residuals to Kleene algebras to obtain \emph{action algebras} \cite{Pratt1991}, and Kozen adds a lattice meet operation to action algebras to obtain \emph{action lattices} \cite{Kozen1994a}; see also \cite{Jipsen2004}. Equational theories of both action algebras and action lattices are undecidable \cite{Kuznetsov2021}.

\begin{definition}
A \emph{right-residuated Kleene algebra with tests} is a structure of the form $(K, B, \cdot, +, \,^{*}, \to,  \,\bar{\,}\, , 1, 0)$ where
\begin{itemize}
\item $(K, B, \cdot, +, \,^{*}, \,\bar{\,}\, , 1, 0) \in \mathsf{KAT}$, and
\item $(K, \cdot, +, \to, 1, 0)$ is a right-residuated idempotent semiring.
\end{itemize}
\end{definition}
 
 We will usually call right-residuated Kleene algebras with tests just residuated. The class of residuated Kleene algebras with tests will be denoted as \textsf{RKAT}.
 
One may be tempted to define $\bar{x}$ as $x \to 0$, but this temptation is better resisted. Note that $0 \in B$ in every Kleene algebra with tests, but $0 \to 0$ is the top element. Hence, if we have a residuated Kleene algebra with tests where $B$ does not contain the top element, then $B$ is not closed under the operation $x \to 0$. An example of such an algebra is given in Figure \ref{fig:res_not_comp}, where $B = \{ 0, 1 \}$ and $+$ is join behaving as indicated in the diagram on the left: $0 \to 0 = \top$ but $\bar{0} = 1$.

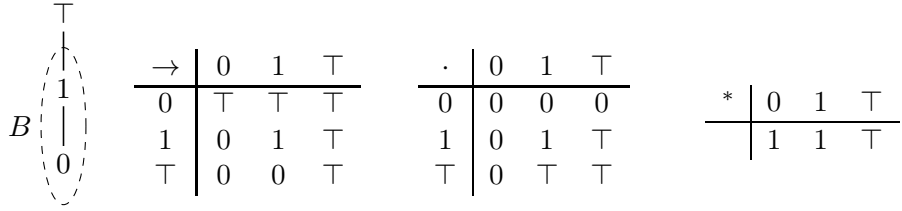
\begin{figure}\centering
\begin{tikzpicture}[every node/.style={inner sep=3pt}]
\node (0) {$0$};
\node[above of=0] (1) {$1$};
\node[above of=1] (2) {$\top$};
\node[draw,dashed,ellipse,fit=(0)(1), label={left:$B$},inner sep=0pt] {};
\draw[semithick] (0) -- (1) -- (2);
\end{tikzpicture}
\quad
\raisebox{1cm}{
\begin{tabular}{c|ccc}
$\to$ & 0 & 1 & $\top$\\
\hline
    0 & $\top$ & $\top$ & $\top$ \\
    1 & 0 & 1 & $\top$ \\
    $\top$ & 0 & 0 & $\top$
\end{tabular} \hspace{.5cm}
\begin{tabular}{c|ccc}
$\cdot$ & 0 & 1 & $\top$\\
\hline
    0 & 0 & 0 & 0 \\
    1 & 0 & 1 & $\top$ \\
    $\top$ & 0 & $\top$ & $\top$
\end{tabular}
}
\hspace{.5cm}
\raisebox{1cm}{
\begin{tabular}{c|ccc}
$^{*}$ & 0 & 1 & $\top$\\
\hline
   & 1 & 1 & $\top$
\end{tabular}
}\caption{A residuated Kleene algebra with tests where $x \to 0$ is not test complementation.}\label{fig:res_not_comp}
\end{figure}

For a similar reason, formulas $p \To b$ of \textsf{S} do not directly translate into $p \to b$; the former entail $1$, but the latter do not, as we have seen. However, a consideration of the relational case sheds some light on the matter. Recall that for $R \subseteq S \times S$ and $B \subseteq \mathrm{id}_S$,
\begin{gather*}
R \to B = \{ (s,u) \mid \forall v.\, (u, v) \in R \implies (s, v) \in B \}\\
R \To B = \{ (u,u) \mid \forall v. \, (u, v) \in R \implies (v,v) \in B\}
\end{gather*}
Now consider two operations $c, e : 2^{S\times S} \to 2^{S\times S}$ such that
\begin{gather}
c (R) = \{ (u, u) \mid \exists s.\, (s, u) \in R \} \label{codomain}\\
e(R) = \{ (s,u) \mid (u,u) \in R \}\label{codomain-adjoint}
\end{gather}
Note that $c$ and $e$ form a \emph{Galois connection} in the sense that, for all binary relations $Q$ and $R$ on $S$,
\begin{equation}\label{Galois}
c (R) \subseteq Q \iff R \subseteq e(Q) \, .
\end{equation}
We have the following observation.
\begin{proposition}\label{prop:teTo}
For all $R \subseteq S \times S$ and $B \subseteq \mathrm{id}_S$,
\begin{equation}\label{teTo}
c (R \to e(B)) = R \To B 
\end{equation}
\end{proposition}
\begin{proof}
If $(u, u) \in R \To B$, then $(u, u) \in R \to e(B)$ since $(v,v) \in B$ iff $(u, v) \in e(B)$ by (\ref{codomain-adjoint}). But then obviously  $(u,u) \in c (R \to e(B))$. Hence, $R \To B \, \subseteq \,  c (R \to e(B))$. Conversely, assume that $(u,u) \in  c (R \to e(B))$ and $(u, v) \in R$, then $(s, u) \in R \to e(B)$ by (\ref{codomain}) and so $(s,v) \in e(B)$. It follows that $(v,v) \in B$. Hence, $(u, u) \in R \To B$ and so $ c (R \to e(B)) \, \subseteq \, R \To B$.
\end{proof}

Hence, it seems that expanding \textsf{RKAT} with abstractions of $c$ and $e$ would make it possible to express the implication operator of \textsf{S}. An abstraction of $\cod$ was added to semirings and Kleene algebras in the works on Kleene algebra with (co)domain \cite{DesharnaisStruth2011,DesharnaisEtAl2006}. We discuss Kleene algebra with codomain in the next section.

\section{Kleene algebra with codomain}\label{sec:KAC}

Kleene algebra with codomain, \textsf{KAC}, is a one-sorted alternative to \textsf{KAT}: instead of assuming the existence of a Boolean subalgebra of tests, \textsf{KAC} introduces a unary \emph{codomain} operation $\cod$ such that the set of $\cod(x)$ for $x$ in the underlying Kleene algebra forms a Boolean algebra. One of the motivations for studying one-sorted alternatives to Kleene algebra with tests is that they seem to be better suited to automated theorem proving (see \cite{DesharnaisStruth2011}, p.~194). Our presentation follows \cite{DesharnaisStruth2011,DesharnaisEtAl2006}, where a symmetric variant of \textsf{KAC}, namely, Kleene algebra with \emph{domain}, was studied.

Kleene algebras with codomain add to Kleene algebras a unary operator that generalizes the \emph{anticodomain} operator on binary relations:
\begin{equation}\label{anticodomain}
a (R) = \{ (u,u) \mid \neg \exists s. (s,u) \in R \} \, .
\end{equation}
Note that, assuming (\ref{codomain}), $c (R) = a (a (R))$.

\begin{definition}
A \emph{Kleene algebra with codomain} is an algebra of the form\\ $(K, \cdot, +, \,^{*}, \ant, 1, 0)$ where $\ant : K \to K$ such that
\begin{gather}
x \cdot \ant(x) \leq 0 \label{ant_ecq}\\
\ant(x \cdot y) \leq \ant \big ( \ant(\ant(x)) \cdot y \big) \label{ant_locality}\\
\ant(x) + \ant(\ant(x)) = 1 \label{ant_lem}
\end{gather}
We define $\cod(x) := \ant(\ant(x))$.
\end{definition}

Note that the relational anticodomain operator (\ref{anticodomain}) satisfies the above equations. 

\begin{proposition}\label{prop:KAC1}
In every Kleene algebra with codomain, the operator $\cod$ satisfies the following codomain equations:
\begin{gather}
x \leq x \cdot \cod(x)\label{cod_preserver}\\
\cod(xy) = \cod(\cod(x)\cdot y) \label{cod_locality}\\
\cod(x) \leq 1\\
\cod(0) = 0\\
\cod(x+y) = \cod(x) + \cod(y) \label{cod_additive}
\end{gather}
Moreover, $\cod$ has the following properties:
\begin{gather}
\cod ( \cod (x)) = \cod (x) \label{cod_fixpoint}\\
\cod (x) \cdot \cod (x) = \cod (x) \label{cod_idempotent}\\
\cod(xy) \leq \cod(y) \label{cod_restriction}\\
x \leq x \cod (y) \implies \cod (x) \leq \cod (y) \label{cod_least_preserver}
\end{gather}
\end{proposition}
\begin{proof}
For the codomain equations, see \cite{DesharnaisStruth2011}, Theorem 8.7, and the more detailed proof of the theorem in Appendix D of \cite{DesharnaisStruth2008}. The arguments there are given for the domain and antidomain operators, but similar arguments work for codomain and anticodomain. 

The fixpoint property (\ref{cod_fixpoint}) follows from $\cod(xy) = \cod(\cod(x)\cdot y)$ and idempotence (\ref{cod_idempotent})  follows from the fixpoint property, $x \leq x \cdot \cod(x)$, and $\cod(x) \leq 1$. The codomain restriction property (\ref{cod_restriction}) is established by noting that $\cod (xy) \leq \cod (\cod (x)y) \leq \cod (1y) \leq \cod (y)$. The least preserver property (\ref{cod_least_preserver}) is established by noting that $x \leq x \cod (y)$ implies $\cod (x) \leq \cod (x \cod (y)) \leq \cod (\cod (y)) \leq \cod (y)$.
\end{proof}

\begin{proposition}\label{prop:KAC2}
Let $\cod(K) = \{ x \mid x \in K \And \exists y. (x = \cod(y))\}$. In every Kleene algebra with codomain,
\begin{enumerate}
\item $(\cod(K), \cdot, +, 1, 0)$ is a subalgebra of $(K, \cdot, +, 1, 0)$;
\item $(\cod(K), \cdot, +, 1, 0)$ is a bounded distributive lattice;
\item $\cod(x) + \ant(\cod(x)) = 1$ and $\cod(x) \cdot \ant(\cod(x)) = 0$.
\end{enumerate}
Hence, $\cod(\mathscr{K}) = (\cod(K), \cdot, +, \ant, 1, 0)$ is a Boolean algebra.
\end{proposition}
\begin{proof}
1.~follows from $\cod (1) = 1$, $\cod (0) = 0$, $\cod (x) + \cod (y) = \cod (x+y)$, and $\cod (x) \cdot \cod(y) = \cod (\cod (x) \cdot \cod (y))$. We verified the latter using Prover 9 \cite{McCune2010}. 2.~follows from 1., together with $\cod (x) \leq 1$ and $\cod (x) = \cod (x) \cod (x)$ of Proposition \ref{prop:KAC1}. 3.~follows from the axioms for $\ant$.
\end{proof}

It can be shown that every full relational \textsf{KAT} can be seen as a \textsf{KAC} where $\ant$ is the relational anticodomain operation, $\ant(R) := \{ (u,u) \mid \neg \exists s. \, (s, u) \in R \}$. Moreover, Proposition \ref{prop:KAC2} shows that every \textsf{KAC} gives rise to a \textsf{KAT}. A more detailed argument involving a translation from the language of \textsf{KAT} into the language of \textsf{KAC} then shows that the equational theory of \textsf{KAT} embeds into the equational theory of \textsf{KAC}. We omit the details.

\section{\textsf{SKAT}: Residuated \textsf{KAC} with a Galois connection}\label{sec:SKAT}

In this section, we introduce \textsf{SKAT}, an extension of \textsf{KAC} with both residuals of Kleene algebra multiplication and the upper adjoint of the codomain operator. The $^{*}$-continuous variant of \textsf{SKAT} will be denoted as $\mathsf{SKAT^{*}}$. The next section shows that \textsf{S} embeds into the equational theory of $\mathsf{SKAT^{*}}$.

\begin{definition}
An \emph{$\mathsf{S}$-type Kleene algebra with codomain} is a structure of the form $(K, \cdot, +, \to, \lto \,^{*}, \ant, \e, 1, 0)$ where 
\begin{itemize}
\item $(K, \cdot, +, \lto, \to, \,^{*}, 1, 0)$ is a residuated Kleene algebra,
\item $(K, \cdot, +, \,^{*}, \ant, 1, 0)$ is a Kleene algebra with codomain, and
\item $\e$ is a unary operation on $K$ that satisfies the following:
\begin{gather}
\ant (\ant (\e (x))) \leq x \label{ante1}\\
x \leq \e (\ant (\ant (x))) \label{ante2}\\
\e (x) \leq \e (x + y) \label{e-mon}
\end{gather}
\end{itemize}
An $\mathsf{S}$-type Kleene algebra with codomain is $^{*}$-continuous iff its underlying Kleene algebra is $^{*}$-continuous. The class of $\mathsf{S}$-type Kleene algebras with codomain will be denoted as \textsf{SKAT} and the class of the $^{*}$-continuous ones as $\mathsf{SKAT}^{*}$.
\end{definition}

We will call $\mathsf{S}$-type Kleene algebras with codomain simply \textsf{SKAT}-algebras. We will usually omit bracketing between elements of $\{ \ant, \e \}$, writing $\e\cod(x)$ instead of $\e ( \cod (x))$, for instance.

Every \textsf{SKAT}-algebra is a Kleene algebra with codomain by definition, and so the test algebra $(\cod(\mathscr{K}), \cdot, +, \ant, 1, 0)$ of each \textsf{SKAT}-algebra $\mathscr{K}$ is a Boolean algebra. It is well known that residuated Kleene algebras form a variety \cite{Pratt1991}, and so \textsf{SKAT} forms a variety.

\begin{proposition}\label{prop:properties_SKAT}
The following hold in each $\mathsf{SKAT}$-algebra:
\begin{gather}
\cod(x) \leq y \iff x \leq \e(y) \label{cod_Galois}\\
\cod (x \to y) \leq x \to x \cod (y) \label{cod_implication}\\
\ant (x) = \cod (\cod (x) \to 0) \label{cod_to_ant} 
\end{gather}
\end{proposition}
\begin{proof}
(\ref{cod_Galois}) follows directly from the \textsf{SKAT} axioms concerning $\e$. (\ref{cod_implication}) is established as follows. Take any $z$ such that $\cod (z) = z$; note that $\cod (x \to y)$ is of this form. We reason as follows:
\begin{itemize}
\item[] $z \leq \cod (x \to y) \implies zx \leq \cod (x \to y)x \overset{(\ref{cod_locality})}{\implies} \cod (zx) \leq \cod ( (x \to y)x)$

\item[] $\implies \cod (zx) \leq \cod (y) \implies zx \cod (zx) \leq \cod(z)x \cod (y) \implies zx \leq x \cod (y)$

\item[] $\implies z \leq x \to x \cod (y) \implies \cod (z) \leq \cod (x \to x \cod (y)) \implies z \leq \cod (x \to x \cod (y))$. 
\end{itemize}

(\ref{cod_to_ant}) is established as follows. Note that $\ant \cod (x) = \ant (x)$ by (\ref{ant_locality}) and $\cod(x) \ant (x) = \ant(x) \cod(x)$ by Proposition \ref{prop:KAC2}. Hence, $\ant (x) \leq \cod (x) \to 0$ by (\ref{ant_ecq}) and residuation. It follows that $\cod\ant(x) \leq \cod (\cod(x) \to 0)$ and so $\ant(x) \leq \cod (\cod (x) \to 0)$.

The converse inequality follows from 
\begin{equation}\label{26}
x \cod (y) = 0 \implies \cod (x) \leq \ant (y)
\end{equation}
In particular, $\cod (\cod (x) \to 0) \cod (x) = 0$ by (\ref{cod_implication}), and so $\cod \cod (\cod (x) \to 0) \leq \ant (x)$ by (\ref{26}), which means that $\cod (\cod (x) \to 0) \leq \ant (x)$ by the fixpoint property (\ref{cod_fixpoint}). The quasi-equation (\ref{26}) is established similarly as the quasi-equation (26) on p.~386 of \cite{DesharnaisStruth2008}. If $x \cod (y) = 0$, then $\ant (\cod (x) \cod (y)) = 1$ and so $\cod (x) \cod (y) = 0$. But then 
$$\cod (x) = \cod (x) (\cod (y) + \ant (y)) = \cod (x) \cod (y) + \cod (x) \ant (y) = \cod (x) \ant (y) \leq \ant (y) \, .$$
\end{proof}

We note that an extension of Kleene algebra with domain with a right residual of Kleene algebra multiplication is considered in \cite{DesharnaisStruth2008,DesharnaisStruth2011}. However, the goal there is to induce a Heyting algebra of tests, and $\to$ is assumed to be a partial operation, defined only on the test algebra. It is also assumed that the test algebra is closed under $\to$, which we may transpose to the setting that uses codomain instead of domain as follows:
\begin{equation}\label{cod_tests_to-closed}
\cod (\cod (x) \to \cod (y)) = \cod (x) \to \cod (y)
\end{equation}
In our setting, $\to$ is a total operation which is not included in the signature of (Boolean) test algebras. It can also be shown that (\ref{cod_tests_to-closed}) does not hold; see Fig.~\ref{fig:cod_test_to-closed_invalid}. Operations $\cdot$ and $^{*}$ are defined as in Fig.~\ref{fig:res_not_comp}, $\lto$ is identical to $\to$. It is easily verified that $\cod (0) \to \cod (0) = \top$, and $\cod (\top) = 1$.

\begin{figure}\centering
\begin{tikzpicture}[every node/.style={inner sep=3pt}]
\node (0) {$0$};
\node[above of=0] (1) {$1$};
\node[above of=1] (2) {$\top$};
\draw[semithick] (0) -- (1) -- (2);
\end{tikzpicture}
\hspace{1cm}
\raisebox{1cm}{
\begin{tabular}{c|ccc}
$\to$ & 0 & 1 & $\top$\\
\hline
    0 & $\top$ & $\top$ & $\top$ \\
    1 & 0 & 1 & $\top$ \\
    $\top$ & 0 & 0 & $\top$
\end{tabular}
}
\hspace{.3cm}
\raisebox{1cm}{
\begin{tabular}{c|ccc}
$\ant$\: & 0 & 1 & $\top$\\
\hline
   & 1 & 0 & $0$
\end{tabular}\hspace{.3cm}
\begin{tabular}{c|ccc}
$\e$\: & 0 & 1 & $\top$\\
\hline
   & 0 & $\top$ & $\top$
\end{tabular}
}\caption{A \textsf{SKAT}-algebra falsifying (\ref{cod_tests_to-closed}).}\label{fig:cod_test_to-closed_invalid}
\end{figure} 

\section{The embedding result}\label{sec:embedding}

In this section we prove that \textsf{S} embeds into the equational theory of \textsf{SKAT}$^{*}$.

Let $\mathsf{X} = \{ \mathsf{x}_i \mid i \in \omega \}$ be a countable set of variables. The set of terms $Tm$ is the set of expressions formed using the following grammar:
$$ p,q := \mathsf{x}_i \mid 1 \mid 0 \mid p \cdot q \mid p + q \mid p \to q \mid p^{*} \mid \cod(p) \mid \e(p)\, .$$ (We abuse the notation by re-using the same variables that range over programs in \textsf{S}, but this will hopefully cause no confusion.) An equation is an ordered pair of terms, denoted as $p \approx q$. We use $p \preceq q$ as short for $p + q \approx q$. If $\mathscr{K} \in \mathsf{SKAT}$, then a \emph{$\mathscr{K}$-valuation}  is any homomorphism from $Tm$ into $\mathscr{K}$. An equation $p \approx q$ \emph{is satisfied} by a $\mathscr{K}$-valuation $\val{\,\,}$ iff $\val{p} = \val{q}$.  The equational theory of \textsf{SKAT} (\textsf{SKAT}$^{*}$) is the set of equations that are satisfied by all $\mathscr{K}$ valuations where $\mathscr{K} \in \mathsf{SKAT}$ ($\mathscr{K} \in \mathsf{SKAT}^{*}$).

\begin{definition}
We define $\tra : Ex_{\mathsf{S}} \to Tm$ as follows:
\begin{itemize}
\item $\tra (\mathsf{p}_n) = \mathsf{x}_{2n}$
\item $\tra (\mathsf{b}_n) = \cod (\mathsf{x}_{2n +1})$
\item $\tra (0) =  \cod(0)$
\item $\tra (b \To c) = \cod (\tra (b) \to \e(\tra (c)))$
\item $\tra (p \oplus q) = \tra (p) + \tra (q)$
\item $\tra (p \otimes q) = \tra (p) \cdot \tra (q)$
\item $\tra (p^{+}) = \tra (p) \cdot \tra (p)^{*}$
\item $\tra (p \To f) = \cod (\tra (p) \to \e(\tra (f)))$
\item $\tra (\epsilon) = 1$
\item $\tra (\Gamma, \Delta) = \tra (\Gamma) \cdot \tra (\Delta)$
\end{itemize}
\end{definition}

\begin{theorem}\label{thm:embedding}
A sequent $\Gamma \vdash f$ is provable in $\mathsf{S}$ iff $\cod \big( \tra (\Gamma)\big) \leq \tra (f)$ belongs to the equational theory of $\mathsf{SKAT}^{*}$.
\end{theorem}

Theorem \ref{thm:embedding} is established in two steps. First we show that every model $M$ of \textsf{S} with a valuation $\val{\:}$ can be transformed into a $\mathscr{M} \in \mathsf{SKAT}^{*}$ and a $\mathscr{M}$-valuation $\vval{\,\,}$ such that $\val{\alpha} = \vval{\tra (\alpha)}$ for all $\alpha \in Ex_{\mathsf{S}}$. Then we show by induction on the length of $\mathsf{S}$-proofs that if $\Gamma \vdash f$ is provable in $\mathsf{S}$, then $\cod \big( \tra (\Gamma)\big) \preceq \tra (f)$ belongs to the equational theory of $\mathsf{SKAT}^{*}$.

\begin{lemma}\label{lem:embedding-1}
Let $M = (W, V)$ be a model of $\mathsf{S}$ and let $\val{\:}$ be a $M$-interpretation. Then there is $\mathscr{M} \in \mathsf{SKAT}^{*}$ and a $\mathscr{M}$-valuation $\vval{\:}$ such that, for all $\alpha \in Ex_{\mathsf{S}}$,
\begin{equation}\label{tra}
\val{\alpha} = \vval{\tra (\alpha)}
\end{equation}
\end{lemma}
\begin{proof}
Fix $M = (W, V)$. Let $\mathscr{M}$ be the full relational residuated Kleene algebra of binary relations on $W$; define $\ant(R) = \{ (u,u) \mid \neg \exists s. \, (s, u) \in R \}$ and $\e(R) = \{ (s,u) \mid (u, u) \in R \}$; note that $\cod (R) = \{ (u,u) \mid \exists s. (s,u) \in R \}$.\footnote{Recall that $R \to Q = \{ (s,u) \mid \forall v.\, (u,v) \in R \implies (s,u) \in Q \}$ and $R \lto Q = \{ (s,u) \mid \forall v.\, (v,s) \in R \implies (v,u) \in Q \}$.} It is easily checked that $\mathscr{M} \in \mathsf{SKAT}^{*}$. In order to prove (\ref{ante1}), note that $(s,t) \in \cod\e (R)$ only if $s = t$ and there is $u$ such that $(u,t) \in \e(R)$, which holds only if $s = t$ and $(t,t) \in R$. In order to prove (\ref{ante2}), note that $(s,t) \in R$ implies $(t,t) \in \cod (R)$, which implies $(s,t) \in \e\cod (R)$. Finally, $\e$ is obviously monotonic, which entails (\ref{e-mon}).

Define $\vval{\:}$ as the unique \textsf{SKAT} homomorphism from $Tm$ to $\mathscr{M}$ satisfying the following, for all $n \in \omega$:
\begin{itemize}
\item $\vval{\mathsf{x}_{2n}} = \val{\mathsf{p}_{n}}$;
\item $\vval{\mathsf{x}_{2n+1}} = \val{\mathsf{b}_{n}}$.
\end{itemize}
We will prove (\ref{tra}) by induction on $\alpha$. First we prove that (\ref{tra}) holds for all tests $b$ by induction on $b$. The base case is established as follows. We know by definition that $\vval{\tra (\mathsf{b}_{n})} = \cod \vval{\mathsf{x}_{2n+1}} = \cod \val{\mathsf{b}_{n}}$. The latter is $\{ (u, u) \mid \exists s. \, (s,u) \in \val{\mathsf{b}_n} \} = \val{\mathsf{b}_{n}}$, since $\val{\mathsf{b}_n} \subseteq \mathrm{id}_W$.  The induction step is established as follows. $\vval{\tra (0)} = \cod \vval{0} =  \emptyset = \val{0}$ and the other case uses (\ref{teTo}) of Proposition \ref{prop:teTo} (we are working with a relational Kleene algebra):
\begin{align*}
\vval{\tra (b \To c)} 
&= \cod ( \vval{\tra (b)} \to \e\vval{\tra (c)}) \\
&= \cod (\val{b} \to \e\val{c})\\
& \overset{\eqref{teTo}}{=} \val{b \To c}
\end{align*}
This shows that (\ref{tra}) holds for all tests $b$. The fact that it holds for all programs is established by easy induction ($\tra$ virtually commutes with Kleene algebra operations). Next, we show that (\ref{tra}) holds for all formulas. The only thing to prove is the induction step for $\alpha = p \To f$. The claim can be established using (\ref{teTo}) as before. The claim for environments is trivial.
\end{proof}

\begin{lemma}\label{lem:embedding-2}
If $\Gamma \vdash f$ is provable in $\mathsf{S}$, then $\cod \big( \tra (\Gamma)\big) \preceq \tra (f)$ belongs to the equational theory of $\mathsf{SKAT}^{*}$.
\end{lemma}
\begin{proof}
We reason by induction on the length of proofs. We note that the proof of this lemma, namely the claim for the (I$^{+}$) rule, is the only point where the assumption of $^{*}$-continuity seems to be required.

 A useful fact to note is that $\tra (f)$ is of the form $\cod(p)$ for all formulas $f$. By a translation of a sequent $\Gamma \vdash f$ we mean the equation $\cod(\tra (\Gamma)) \preceq \tra (f)$. We will prove that translations of axioms of $\mathsf{S}$ are in the equational theory of $\mathsf{SKAT}^{*}$, and that if translations of all premises of an inference rule of $\mathsf{S}$ are in the equational theory of $\mathsf{SKAT}^{*}$, then so is the consequence of the rule. (If $p \preceq q$ is in the equational theory of $\mathsf{SKAT}^{*}$, then we say that $p \preceq q$ is valid.)
 
\begin{itemize}
\item[(Id)] $\cod (\tra (b)) \preceq \tra (b)$ is valid since $\cod\cod (x) \leq \cod (x)$ by (\ref{cod_fixpoint}).\\

\item[(I0)] $ \cod(\tra (\Gamma) \cdot \tra(0) \cdot \tra(\Delta)) \preceq \tra (f)$ is valid since $\tra (0) = \cod(0) = 0$.\\

\item[(TC)] $\dfrac{\cod\big(\tra (\Gamma) \tra (b) \tra (\Delta) \big) \preceq \tra (f) \quad \cod \big( \tra (\Gamma) \tra (b \To 0) \tra (\Delta) \big) \preceq \tra (f)}{\cod \big(\tra (\Gamma) \tra (\Delta) \big) \preceq \tra (f)}$\\[1mm] preserves validity thanks to (\ref{cod_additive}) and (\ref{cod_to_ant}):
\begin{itemize}\setlength{\itemsep}{2mm}
\item[] $\cod (x \cod(y) z) \leq a$ and  $\cod (x \cod (\cod(y) \to 0) z) \leq a$
\item[] $\overset{(\ref{cod_additive})}{\implies}$\: $ \cod \big( x \cod(y) z \,+\, x \cod (\cod(y) \to 0) z \big) \leq a$
\item[] $\overset{(\ref{cod_to_ant})}{\implies}$\:
	$\cod \big ( x \cod (y) z \, +\, x \ant (y) z \big) \leq a$
\item[] $\implies$\: 
	$\cod \big ( x \big ( \cod(y) + \ant (y) \big) z\big) \leq a$
\item[] $\overset{(\ref{ant_lem})}{\implies}$\:
	$\cod (xz) \leq a$	 	
\end{itemize}

\medskip

\item[(cut)] $\dfrac{\cod (\tra (\Gamma)) \preceq \tra (g) \qquad \cod \big( \tra (\Gamma) \tra (g) \tra (\Delta) \big) \preceq \tra (f)}{\cod \big( \tra (\Gamma) \tra (\Delta) \big) \preceq \tra (f)}$\\[1mm] preserves validity thanks to (\ref{cod_preserver}) and (\ref{cod_additive}):
\begin{itemize}\setlength{\itemsep}{2mm}
\item[] $\cod(x) \leq a_2 $ and $ \cod( x a_2 y) \leq a_1$
\item[] $\implies$\:  
	$x \cod(x) y \leq x a_2 y$ and $ \cod( x a_2 y) \leq a_1$
\item[] $\overset{(\ref{cod_preserver})}{\implies}$\:
	$x y \leq x a_2 y$ and $ \cod( x a_2 y) \leq a_1$
\item[] $\overset{(\ref{cod_additive})}{\implies}$\:
	$\cod(x y) \leq \cod (x a_2 y)$ and $ \cod( x a_2 y) \leq a_1$
\item[] $\implies$\:
	$\cod (xy) \leq a_1$	
\end{itemize}

\medskip

\item[(R$\To$)] $\dfrac{\cod (\tra (\Gamma) \tra (p)) \preceq \tra(f)}{\cod(\tra (\Gamma)) \preceq \tra (p \To f)}$ preserves validity thanks to (\ref{cod_Galois}) and (\ref{cod_additive}):
\begin{itemize}\setlength{\itemsep}{2mm}
\item[]  $\cod(xy) \leq z$\:
 $\overset{(\ref{cod_Galois})}{\implies}$\:
 $xy \leq \e (z)$
 \item[] $\implies$\: $x \leq y \to \e (z)$ \:
 $\overset{(\ref{cod_additive})}{\implies}$\:
 $\cod (x) \leq \cod (y \to \e (z))$
\end{itemize}

\medskip

\item[(I$\To$)] $\dfrac{\cod \big( \tra (\Gamma) \tra (p) \tra (f) \tra (\Delta) \big) \preceq \tra (g) }{\cod \big ( \tra (\Gamma) \tra (p \To f) \tra (p) \tra (\Delta) \big ) \preceq \tra(g) }$ preserves validity thanks to (\ref{cod_implication}) and the fact that (\ref{cod_Galois}) entails $\cod \e \cod (y) \leq \cod(y)$:
\begin{gather*}
\cod (x \to \e\cod(y)) x \,\leq\, (x \to x\, \cod\e\cod(y))x \,\leq\, x\, \cod\e\cod(y) \,\leq\, x \cod(y)
\end{gather*}
(Recall that $\tra(p \To f) = \cod (q \to \e\cod (r))$ for some terms $q,r$.)\\

\item[(I$\otimes)$] follows directly from the definition of $\tra$, and so do (E$\otimes$), (I$\oplus$), (E$\oplus_1$) and (E$\oplus_2$).\\ 

\item[(I$^{+}$)] $\dfrac{\cod (\tra (g) \tra (p)) \preceq \tra (f) \quad \cod (\tra (g) \tra (p)) \preceq \tra (g)}{\cod \big (\tra (g) \tra (p) \tra (p)^{*}\big) \preceq \tra (f)}$.\\[1mm] It can be proven by induction on $n$, using (\ref{cod_locality}), that if $\cod(yx) \leq \cod(z)$ and $\cod(yx) \leq y$, then $\forall n \geq 0. \, \cod (yxx^{n}) \leq \cod(z)$. By (\ref{cod_Galois}), $\forall n \geq 0. \, yxx^{n} \leq \e \cod(z)$, and so by $^{*}$-continuity $yxx^{*} \leq \e \cod(z)$, which means that $\cod(yxx^{*}) \leq \cod(z)$. (Note that $^{*}$-continuity entails that infinite suprema of elements of the form $xy^{*}z$ exist.) In fact, this is the only point where we need to assume $^{*}$-continuity.\\

\item[(E$^{+}$)] follows from the fact that $x \leq xx^{*}$ in Kleene algebras; and (CC$^{+}$) follows from the fact that $x^{*}xx^{*} \leq x^{*}$.\\

\item[(W$f$)] $\dfrac{\cod (\tra (\Gamma) \tra (\Delta)) \preceq \tra (g)}{\cod (\tra (\Gamma) \tra (f) \tra (\Delta)) \preceq \tra (g)}$ preserves validity thanks to $\cod(x) \leq 1$.\\

\item[(W$p$)] $\dfrac{\cod (\tra (\Gamma)) \preceq \tra (f)}{\cod (\tra (p) \tra (\Gamma)) \preceq \tra (f)}$ follows from (\ref{cod_restriction}).
\end{itemize}
This concludes the proof of the lemma.
\end{proof}

Now we are in the position to prove Theorem \ref{thm:embedding}.

\begin{proof}
If $\Gamma \vdash f$ is not provable in \textsf{S}, then there is a model $M$ and an interpretation function $\val{\:}$ such that $(w,u) \in \val{\Gamma}$ and $(u,u) \notin \val{f}$ by Kozen and Tiuryn's relational completeness theorem. By Lemma \ref{lem:embedding-1}, there is $\mathscr{M} \in \mathsf{SKAT}^{*}$ and a $\mathscr{M}$-valuation $\vval{\:}$ such that $(u,u) \in \vval{\cod (\tra (\Gamma))}$ and $(u,u) \notin \vval{\tra (f)}$. Hence, $\cod ( \tra (\Gamma)) \preceq \tra (f) $ is not in the equational theory of \textsf{SKAT}$^{*}$. Conversely, if $\Gamma \vdash f$ is provable in \textsf{S}, then $\cod ( \tra (\Gamma)) \preceq \tra (f) $ is in the equational theory of \textsf{SKAT}$^{*}$ by Lemma \ref{lem:embedding-2}.
\end{proof}

\section{Conclusion}\label{sec:conclusion}

We have shown in this paper that Kozen and Tiuryn's substructural logic of partial correctness \textsf{S} embeds into the equational theory of an extension of $^{*}$-continuous Kleene algebra with codomain with both residuals of the Kleene algebra multiplication and the upper adjoint of the codomain operator, $\mathsf{SKAT}^{*}$. We believe that this result sheds more light on the landscape of program logics and algebras.

A number of interesting problems remain open:
\begin{enumerate}
\item Does $\mathsf{S}$ embed into the equational theory of $\mathsf{SKAT}$?

\item Which residuated Kleene algebras extend to algebras in \textsf{SKAT} or \textsf{SKAT}$^{*}$? Is the equational theory of \textsf{SKAT} (\textsf{SKAT}$^{*}$) a conservative extension of the equational theory of residuated ($^{*}$-continuous) Kleene algebras?

\item Related to the previous problem, is the equational theory of $\mathsf{SKAT}^{*}$ identical to the equational theory of $\mathsf{SKAT}$?  The equational theory of $^{*}$-continuous residuated Kleene algebras is not identical to the equational theory of residuated Kleene algebras \cite{Buszkowski2006a}.

\item What are the free algebras in $\mathsf{SKAT}$ and $\mathsf{SKAT}^{*}$? McLean \cite{McLean2020} shows that regular sets of labelled pointed trees form free Kleene algebras with domain, and a version of his construction may turn out to apply to our case as well.

\item What is the complexity of the equational theories of $\mathsf{SKAT}$ and $\mathsf{SKAT}^{*}$? The equational theory of $^{*}$-continuous residuated Kleene algebras is $\Pi^{0}_1$-complete (hence not r.e.) and the equational theory of residuated Kleene algebras is $\Sigma^{0}_1$-complete.

\item Consider the set of terms generated by the following restricted grammar:
$$p, q := \mathsf{x}_i \mid 1 \mid 0 \mid p \cdot q \mid p + q \mid p^{*} \mid \cod(p) \mid \cod(p \to \e \cod(q)) \, .$$
Is the fragment of the equational theory of $\mathsf{SKAT}$ restricted to terms generated by this grammar decidable?

\item What are the Gentzen systems sound and (weakly) complete for $\mathsf{SKAT}$, in the style of \cite{Jipsen2004}?
\end{enumerate}

\subsubsection*{Acknowledgement}
We are grateful to two anonymous reviewers for valuable suggestions. The first author is grateful to V\'{i}t Pun\v{c}och\'{a}\v{r} for discussions that led to the work on this paper, and to the audience at the Proof Theory Seminar of the Steklov Mathematical Institute for illuminating discussion after his talk on 14 February 2022. The work of the second author was carried out within the project \textit{Supporting the Internationalization of the Institute of Computer Science of the Czech Academy of Sciences} (no.\ CZ.02.2.69/0.0/0.0/18\_053/0017594), funded by the Operational Programme Research, Development and Education of the Ministry of Education, Youth and Sports of the Czech Republic. The project is co-funded by the EU.

\end{document}